\documentclass[10pt,a4paper]{article}
\pdfoutput=1
\usepackage[utf8]{inputenc}
\usepackage{amsmath}
\usepackage{amsfonts}
\usepackage{amssymb}
\usepackage[thmmarks,amsmath]{ntheorem}
{\theoremstyle{nonumberplain}
	\theoremheaderfont{\itshape}
	\theorembodyfont{\normalfont}
	\theoremsymbol{\ensuremath\square}
	\newtheorem{proof}{Proof.}}
\usepackage{breqn}
\usepackage{flexisym}
\usepackage{mathstyle}
\usepackage{color}
\usepackage{authblk}
\usepackage{multicol}
\usepackage{enumitem}

\usepackage{graphicx}
\usepackage{subcaption}

\usepackage[colorlinks]{hyperref}
\usepackage[nameinlink]{cleveref}

\newtheorem{theorem}{Theorem}[section]

\newtheorem{lemma}[theorem]{Lemma}

\newtheorem{conjecture}[theorem]{Conjecture}
\newtheorem{corollary}[theorem]{Corollary}

\author{Caishi Fang\thanks{thinliber@gmail.com} }
\affil{Centre for Quantum Software and Information,\\
	University of Technology Sydney, Australia}
\title{Accelerations for Graph Isomorphism }
\date{\vspace{-5ex}}
\begin{document}

\maketitle

\begin{abstract}
In this paper, we present two main results. First, by only one conjecture (\Cref{sym-pat}) for recognizing a vertex symmetric graph, which is the hardest task for our problem, we construct an algorithm for finding an isomorphism between two graphs in polynomial time $ O(n^{3}) $. Second, without that conjecture, we prove the algorithm to be of quasi-polynomial time $ O(n^{1.5\log n}) $. The conjectures in this paper are correct for all graphs of size no larger than $ 5 $ and all graphs we have encountered. At least the conjecture for determining if a graph is vertex symmetric is quite true intuitively. We are not able to prove them by hand, so we have planned to find possible counterexamples by a computer. We also introduce new concepts like collapse pattern and collapse tomography, which play important roles in our algorithms. 
\end{abstract}

\section{Introduction}
Currently, the best general algorithm for graph isomorphism problem is due to Babai \cite{babai2016graph}, who shows that the graph isomorphism is of quasi-polynomial time $ \exp((\log n)^{O(1)}) $. We give a constructive proof of this result in the current paper. And by only one conjecture, which is quite true intuitively, we give a polynomial time algorithm for the problem. This also means that we have reduced the graph isomorphism problem in polynomial time to the problem of determining whether a graph is vertex symmetric or not.

A detailed review is needed on the applications and related problems such as group isomorphism. The author has not seen Babai's \cite{babai2016graph} work in very detail. There may be some common techniques between this paper and previous papers not pointed out, which is another reason for the need of a review paper. Currently, if you want to know more about the origin and research history of the graph isomorphism problem, please refer to \cite{miller1979graph,bunke2000graph,babai2016graph,kobler2012graph}.\\

In this paper, $ G=(V,E) $ means an undirected graph $ G $ with a vertex set $ V $ and an edge set $ E $ and without self-loop or multiple edges connecting two vertexes. The case that graphs containing self-loop and multi-edges is discussed after we have presented the main results. Let $ V=\{v_1, v_2, \cdots, v_n\} $. If there is a direct connecting between $ v_i $ and $ v_j $ for $ v_i, v_j\in V $ and $ v_j\neq v_i $, we denote it by $ (v_i, v_j) $. Let $ T $ be the set of all $ (i,j) $ pairs with $ (v_i, v_j) $. $ E=\{(v_i,v_j):(i,j)\in T \} $. We denote the size of $ V $ as $ |V| $. Let $ V_j=\{v_k\in V: (v_j, v_k)\in E \} $. The \emph{degree of a vertex} $ v_j\in V $ is $ |V_j| $.

Given two graphs $ G_1=(V_1, E_1) $ and $ G_2=(V_2, E_2) $, if there is a one-to-one correspondence $ \pi $ between $ V_1 $ and $ V_2 $, s.t., for all $ v_i $, $ v_j\in V_1 $, $ (v_i,v_j)\in E_1 $ iff $ (\pi(v_i), \pi(v_j))\in E_2 $, then we say $ G_1 $ is \emph{isomorphic} to $ G_2 $ and $ \pi $ is a graph isomorphism between $ G_1 $ and $ G_2 $. Any isomorphism from a graph $ G $ to itself is called the \emph{automorphism} of $ G $. Note that, as the identity permutation is always an automorphism for any graph, we are not interested in this trivial automorphism. For the graph isomorphism and automorphism, we have an intuitive understanding, i.e., all directly connected vertexes must be also directly connected after an isomorphic or automorphic mapping, which is a permutation of vertex names.

Now we define several problems.

The \emph{graph isomorphism problem}, denoted as $ GI(G_1, G_2) $: determine whether $ G_1 $ and $ G_2 $ are isomorphic.

The \emph{graph automorphism problem}, denoted as $ GA(G) $: determine if there exists a non-trivial automorphism of the undirected graph $ G $.

The \emph{graph automorphism counting problem}, denoted as $ \# GA(G) $: find the total number of automorphisms of the undirected graph $ G $.

If we can solve a problem $ P $ by using polynomially many times of the procedure for solving another problem $ Q $, we say $ P $ is \emph{polynomially reducible} to $ Q $. If $ P $ and $ Q $ are polynomially reducible to each other, we say they are \emph{polynomially equivalent}.

It is shown $ GI(G_1, G_2) $ is polynomially equivalent with $ \# GA(G) $ \cite{mathon1979note,kobler1993graph}, and $ GA(G) $ is polynomially reducible to $ GI(G_1, G_2) $ \cite{kobler1993graph}.

It seems $ \# GA(G) $ is much harder than $ GA(G) $. Just take the complete graph as an example, the total number of automorphisms is $ n!-1 $, which is hard to find one by one. With an oracle for $ GI(G_1, G_2) $, it will be easy to do $ \# GA(G) $.

In this paper, we introduce two more problems. One is the \emph{graph automorphism with constraint problem}, denoted as $ GA(G, C) $: determine if there exists a non-trivial automorphism of the undirected graph $ G $, with the constraint $ C $. For $ v_1 $, $ v_2\in G $, $ (v_1, v_2)\in C $ if $ v_1 $ and $ v_2 $ cannot replace each other under any permutation or $ \langle v_1,v_2\rangle\in C $ if $ v_1 $ and $ v_2 $ can only correspond to each other in any automorphic mapping. 

The other is the \emph{graph isomorphism with constraint problem}, denoted as $ GI(G_1,G_2, C) $: determine whether $ G_1 $ and $ G_2 $ are isomorphic, with the constraint $ C $. For $ v_1\in G_1 $, $ v_2\in G_2 $, $ (v_1, v_2)\in C $ if $ v_1 $ and $ v_2 $ cannot replace each other under any mapping or $ \langle v_1,v_2\rangle\in C $ if $ v_1 $ and $ v_2 $ can only correspond to each other in any isomorphic mapping.

It is easy to see that (1) $ GI(G_1,G_2) $ is a special case of $ GI(G_1,G_2,C) $; (2) $ GA(G) $ is a special case of $ GA(G, C) $; (3) $ GI(G_1,G_2) $ is equivalent to $ GA(G_1\cup G_2, C) $, with $ C=\{ \langle v_1,v_2\rangle: v_1, v_2\in G_1\text{ or }v_1, v_2\in G_2 \} $; (4) $ GA(G, C) $ seems easier than $ GA(G) $, as some permutations are ruled out by the constraint $ C $. In our polynomial algorithm for $ GI(G_1,G_2) $, we make use of $ GI(G_1,G_2,C) $ as a subroutine.

In \cite{mathon1979note}, it is said that the checking and counting of graph isomorphism are polynomially equivalent, which is an evidence to the conjecture that the graph isomorphism is not $ \mathbf{NP} $-complete.

In \cite{booth1979problems,kobler1993graph}, they define a complexity class $ \mathbf{GI} $ of all problems polynomially reducible to the graph isomorphism problem, and claim that $ \mathbf{GI}=\mathbf{P} $ if the graph isomorphism is in class $ \mathbf{P} $. 

We introduce concepts and theoretic work, including our conjectures, in Section \ref{prepa}. The algorithm for graph isomorphism is described in Section \ref{algo}. The final section is the conclusion. The reader may skip \Cref{var-pat} if not interested in too much theoretical work.

\section{Preparations}\label{prepa}

Given two graphs $ G=(V,E) $ and $ G'=(V',E') $, with $ |V|=|V'| $, our target is to find a one-to-one correspondence $ \pi: V\rightarrow V' $, s.t., \[ \text{for all }v_i,\ v_j\in V,\ (v_i,v_j)\in E\text{ iff }(\pi(v_i),\pi(v_j))\in E'. \]

In this section, $ V=\{v_1,v_2,\cdots, v_n \} $, $ V'=\{v'_1,v'_2,\cdots,v'_n \} $; $ V_i=\{v_j\in V: (v_i, v_j)\in E \} $, $ V'_i=\{v'_j\in V': (v'_i, v'_j)\in E' \} $, for $ i=1,\ 2,\ \cdots,\ n $. 

The first information we can use is that $ |V_i|=|V'_j| $ if $ \pi $ is an isomorphism and $ \pi(v_i)=v'_j $, i.e., in any isomorphism between $ V $ and $ V' $, a vertex of $ V $ can only be mapped to a vertex of $ V' $ with the same degree. For this reason, we introduce a concept called \emph{base subgraph} $ G^{(w)} $ for those vertexes of the same degree $ w $ in a graph $ G $: 
\[ G^{(w)}=(V^{(w)},E^{(w)}),\ V^{(w)}=\{v_i\in V: |V_i|=w \}, \] \[ E^{(w)}=\{(v_i,v_j)\in E: v_i, v_j\in V^{(w)} \}. \]
Note that it is only possible to map a $ v_i\in V^{(w)} $ to some $ v'_j\in {V'}^{(w)} $.

Given $ G^{(w)} $ and $ v_i\in V^{(w)} $, we define the \emph{extension} based on $ v_i $, \[ G^{ex}(v_i,w)=(V^{ex}(v_i,w),E^{ex}(v_i,w)), \] as follows: 
\begin{enumerate}	
	\item $ v_i\in V^{ex}(v_i,w) $; 
	\item $ V_i-V^{(w)}\subseteq V^{ex}(v_i,w) $;
	\item For any $ v_j\in V^{ex}(v_i,w)-V^{(w)} $, $ V_j\subseteq V^{ex}(v_i,w) $;
	\item $ E^{ex}(v_i,w)=\{(v_j,v_k)\in E: v_j\in V^{ex}(v_i,w)-V^{(w)}, v_k\in V^{ex}(v_i,w) \} $.
	
\end{enumerate} 

We call $ v_i $ the \emph{base point} of the extension $ G^{ex}(v_i,w) $.

The base subgraph $ G^{(w)} $ is a separation of the the graph $ G $, which means vertexes inside the base subgraph are different from those outside. As such separation is not limited to the degree argument, we can generalize the concepts, base subgraph and extension, to a given set of vertexes $ \beta\subsetneq V $. The base subgraph of $ \beta $ is $ G^{\beta}=(V^{\beta},E^{\beta}) $, where $ V^{\beta}=\beta $. The extension of $ G^{\beta} $ based on $ v_i $ is $ G^{ex}(v_i,\beta) $. We just replace $ V^{(w)} $ by $ \beta $ in the definition of $ G^{(w)} $ and $ G^{ex}(v_i,w) $.

As the definitions are not so intuitive, we explain them by an example. Let's consider the graph in Figure \ref{fig:1}. The base subgraph of degree $ 3 $ is depicted in Figure \ref{fig:g005}. All extensions with respect to this base subgraph are depicted in Figure \ref{fig:g011}, \ref{fig:g012}. Note we have colored all base points in black.

\begin{figure}[h]
	\centering
	\begin{subfigure}[b]{0.4\textwidth}
		\centering
		\includegraphics[width=0.9\textwidth]{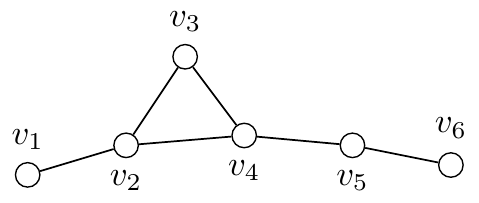}
		\caption{}
		\label{fig:1}
	\end{subfigure}
	\begin{subfigure}[b]{0.25\textwidth}
		\centering
		\includegraphics[width=0.5\textwidth]{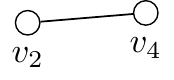}
		\caption{}
		\label{fig:g005}
	\end{subfigure}
	\begin{subfigure}[b]{0.25\textwidth}
		\centering
		\includegraphics[width=0.7\textwidth]{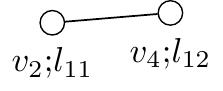}
		\caption{}
		\label{fig:g013}
	\end{subfigure}	
	
	\begin{subfigure}[b]{0.4\textwidth}
		\centering
		\includegraphics[width=0.5\textwidth]{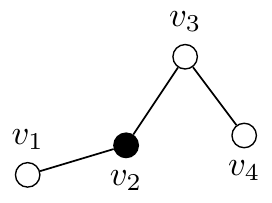}
		\caption{}
		\label{fig:g011}
	\end{subfigure}
	\begin{subfigure}[b]{0.4\textwidth}
		\centering
		\includegraphics[width=0.7\textwidth]{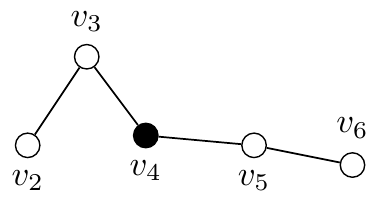}
		\caption{}
		\label{fig:g012}
	\end{subfigure}
	
	\caption{(a) $ G_a $; (b) $ G_a^{(3)} $; (c) $ G_a^{(3)} $ with labels; (d) $ G_a^{ex}(v_2,3) $; (e) $ G_a^{ex}(v_4,3) $ }
	\label{fig:2}
\end{figure}

Now we introduce some sets of \emph{labels} $ L_k=\{l_{k1}, l_{k2}, l_{k3}, \cdots \} $, for $ k=1,2,\cdots $. The total number of labels used in an algorithm will be finite. We use labels to replace extensions with respect to a base subgraph. The labels serve as carrying on the information of whether two extensions are isomorphic or not. If two extensions are isomorphic, we replace them with the same label, otherwise different labels. After pair-wise comparing Figure \ref{fig:g011}, \ref{fig:g012}, we find that they are not isomorphic with the constraint that a base point (black point) can only be mapped to a base point, so we label them differently. In Figure \ref{fig:g013}, we focus on the base subgraph of degree 3, the base point $ v_2 $ now has a label $ l_{11} $ and the base point $ v_4 $ has a label $ l_{12} $. It is better to write labels of a base point in a predefined order and combine identical labels, e.g., write $ l_{23} l_{11} l_{12} l_{12} $ as $ l_{11} l_{12}^2 l_{23} $. We will learn more about the labeling procedure in the algorithm for graph isomorphism.  

We say a graph is \emph{vertex regular} of $ w $ if every vertex in the graph is of the same degree $ w $. 

Given a graph $ G=(V,E) $, with $ v_i $ and $ V_i $ for $ i= 1, 2, \cdots, n $, we define the \emph{collapse} of $ G $ with the \emph{trigger} $ v_k $, $ G^{col}(v_k) $, as follows:
\begin{enumerate}
	\item Layer $ 0 $: $ G^{col}(v_k,0) $, $ V^{col}(v_k,0)=\{v_k\} $, $ E^{col}(v_k,0)=\varnothing $;
	\item Layer $ 1 $: $ G^{col}(v_k,1) $, $ V^{col}(v_k,1)=V_k $, $ E^{col}(v_k,1)=\{(v_i,v_j)\in E: v_i,v_j\in V_k \} $;
	\item Layer $ i+2 $: Given Layer $ i+1 $ ($ G^{col}(v_k,i+1) $) and Layer $ i $ ($ G^{col}(v_k,i) $), Layer $ i+2 $ is $ G^{col}(v_k,i+2) $, with $ V^{col}(v_k,i+2)=\bigcup_{v_x\in V^{col}(v_k,i+1) }\{v_y\in V: v_y\in V_x-V^{col}(v_k,i+1)-V^{col}(v_k,i) \} $, $ E^{col}(v_k,i+2)=\{(v_x,v_y)\in E: v_x,v_y\in V^{col}(v_k,i+2) \} $.
\end{enumerate}
See Figure \ref{fig:g016} for example. It is a collapse trigged by vertex $ v_1 $ for the vertex regular graph $ G_b $, which is Figure \ref{fig:g014}. There are $ n $ collapses in a graph of $ n $ vertexes.

\begin{figure}[h]
	\centering
	\begin{subfigure}[b]{0.4\textwidth}
		\centering
		\includegraphics[width=0.8\textwidth]{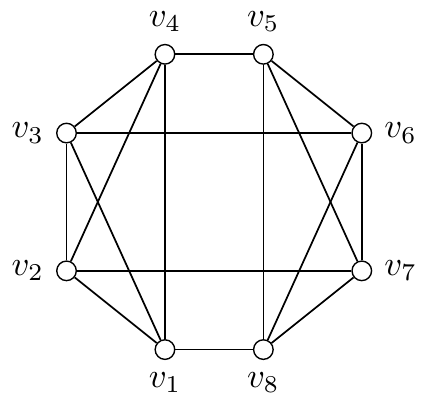}
		\caption{}
		\label{fig:g014}
	\end{subfigure}
	\begin{subfigure}[b]{0.5\textwidth}
		\centering
		\includegraphics[width=0.7\textwidth]{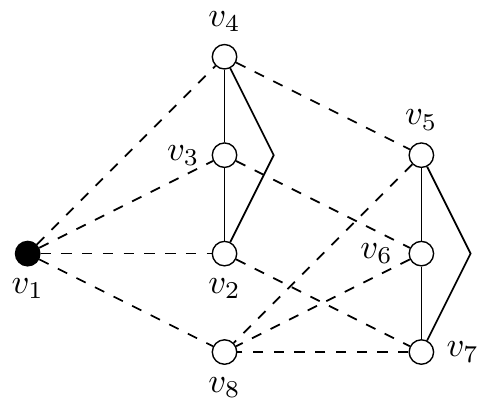}
		\caption{}
		\label{fig:g016}
	\end{subfigure}
	
	\caption{(a) vertex regular graph $ G_b $; (b) $ G_b^{col}(v_1) $}
	\label{fig:3}
\end{figure}

Let $ deg(v_i,v_j)=|V_i|+|V_j|-|V_i\cap V_j| $ be the degree (\emph{edge degree}) of an edge $ (v_i,v_j) $. Intuitively, the degree of an edge is the total number of collapses in which the edge is located before Layer $ 2 $. Likewise, we say a graph is \emph{edge regular} of $ w $ if every edge in the graph is of the same degree $ w $. Figure \ref{fig:g017} and Figure \ref{fig:g018} are both edge regular and vertex regular, and they are actually isomorphic. Figure \ref{fig:g014} is vertex regular but not edge regular.

\begin{figure}[h]
	\centering
	\begin{subfigure}[b]{0.49\textwidth}
		\centering
		\includegraphics[width=0.9\textwidth]{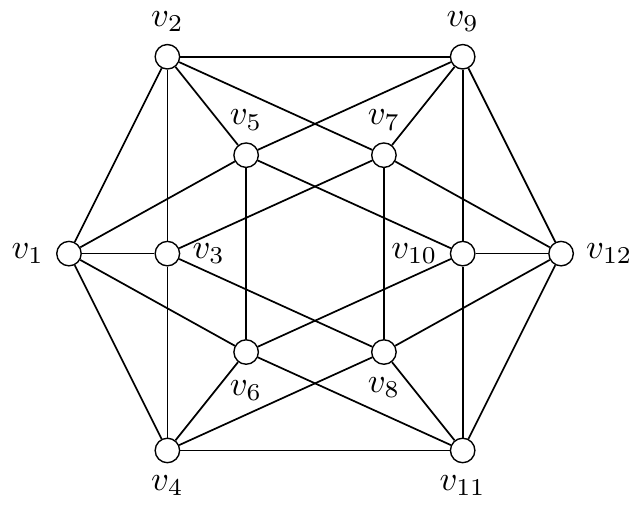}
		\caption{}
		\label{fig:g017}
	\end{subfigure}
	\begin{subfigure}[b]{0.49\textwidth}
		\centering
		\includegraphics[width=0.9\textwidth]{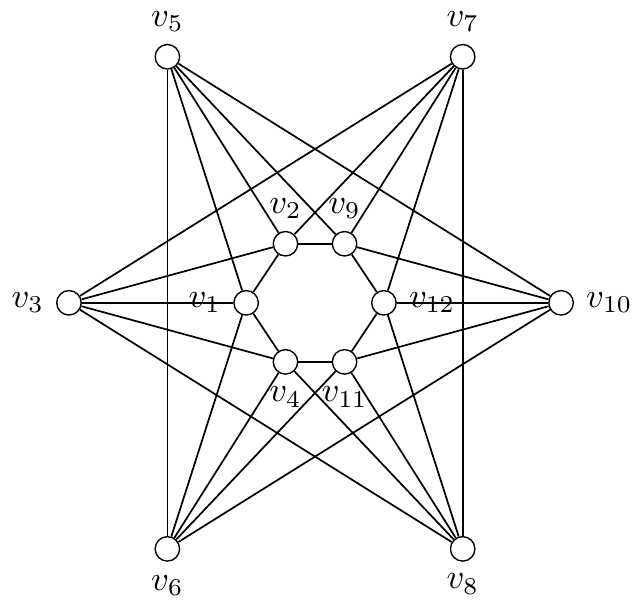}
		\caption{}
		\label{fig:g018}
	\end{subfigure}
	
	\caption{(a) graph $ G_c $; (b) graph $ G_d $}
	\label{fig:4}
\end{figure}

Let's define here an important class of graphs, the vertex symmetric graph. For a graph $ G=(V,E) $, by designating a \emph{nailed vertex} (imagine that you nail the graph to the wall with the nail on the designated vertex), we distinguish $ |V| $ possible \emph{nailed graphs} of $ G $, $ G(v_i) $, $ i=1,2,\cdots,|V| $. Those independent parts of the graph will fall down to the floor. Thus, $ G(v_i) $ is equal to $ G^{col}(v_i) $ with the edges between its layers added back. Any nailed graph is connected. $ G $ is \emph{vertex symmetric} if and only if all of its nailed graphs are isomorphic to each other. More precisely, if for every $ i $ from $ 1 $ to $ |V| $, we have for every $ j $ from $ i $ to $ |V| $, $ GI(G(v_i),G(v_j),\{\langle v_i,v_j \rangle \})=True $, then we say the graph $ G $ is vertex symmetric. Here, $ \{\langle v_i,v_j\rangle \} $ means $ v_i $ can only correspond to $ v_j $ in any isomorphic mapping and vice versa. Complete graph is vertex symmetric. A graph consisting of one circle or multiple identical circles is also vertex symmetric. Figure \ref{fig:g017} and Figure \ref{fig:g018} are further examples of vertex symmetric graphs. Later we will see that the only obstacle of the graph isomorphism problem is how to recognize a vertex symmetric graph.

\begin{lemma}\label{sym-vertex}
	Given $ G=(V,E) $, if $ G $ is vertex symmetric, then it is vertex regular.
\end{lemma}

We can generalize the concepts, collapse and nailed graph, to multi-collapse and multi-nailed graph, i.e., using multiple triggers or nailed vertexes at the same time. Given a graph $ G $, a \emph{multi-collapse} (sometimes just called collapse if it is clear from the content) with a set of triggers $ T $ is $ G^{col}(T) $, in which layer $ 0 $ contains vertexes in the set $ T $, and layer $ k+1 $ contains all vertexes connecting those vertexes in layer $ k $. A \emph{multi-nailed graph} with a set of nailed vertexes $ T $ is $ G(T) $, which looks the same as $ G^{col}(T) $ with those edges between layers added back. The normal (not nailed) graph and the nailed graph are special cases of the multi-nailed graph.

A graph $ G $ is \emph{edge symmetric} if and only if all of its multi-nailed graphs $ G((v_i,v_j)) $, with $ (v_i, v_j)\in E $, are isomorphic to each other under the constraint $ \{\langle (v_i,v_j), (v_{i'}, v_{j'})\rangle \} $, which means that the edge $ (v_i,v_j) $ can only be mapped to $ (v_{i'}, v_{j'}) $ in any isomorphism between $ G((v_i,v_j)) $ and $ G((v_{i'}, v_{j'})) $.

\begin{lemma}\label{sym-edge}
	Given $ G=(V,E) $, if $ G $ is edge symmetric, then it is edge regular.
\end{lemma}

Actually, there exist graphs that is both vertex regular and edge regular but not vertex symmetric, see Figure \ref{fig:g019}. There exist graphs that is edge symmetric but not vertex symmetric, e.g., Figure \ref{fig:g019}, the complete bipartite graph in Figure \ref{fig:g021} and Figure 3.2 in \cite{godsil2001algebraic}. There exist graphs that is vertex symmetric but not edge symmetric, see Figure \ref{fig:g020}. There exist graphs edge symmetric but not vertex regular, see Figure \ref{fig:g021}. There exists graphs that is both vertex regular and edge regular but not vertex symmetric or edge symmetric, e.g., a graph consisting of two independent circles of different sizes.

\begin{figure}[h]
	\centering
	\begin{subfigure}[b]{0.49\textwidth}
		\centering
		\includegraphics[width=0.95\textwidth]{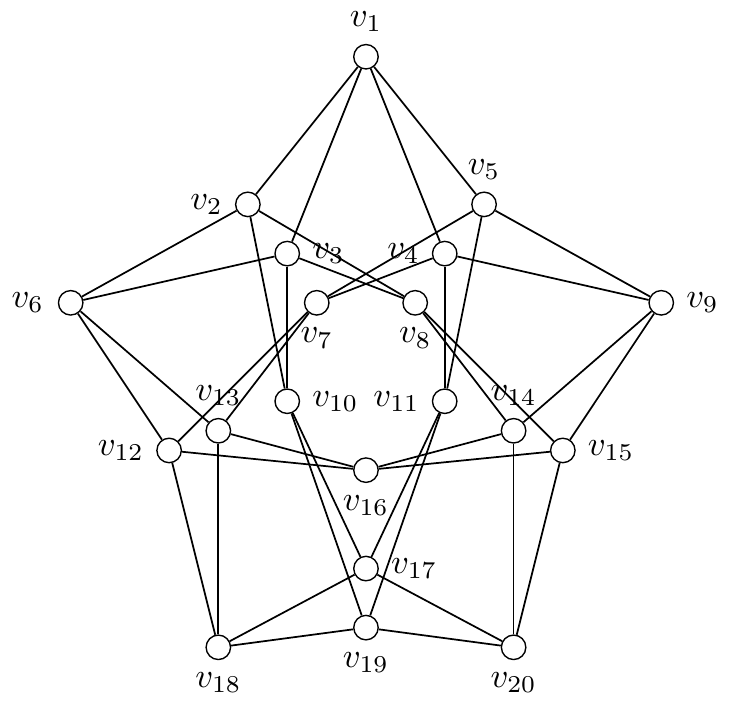}
		\caption{}
		\label{fig:g019}
	\end{subfigure}
	\begin{subfigure}[b]{0.3\textwidth}
		\centering
		\includegraphics[width=0.6\textwidth]{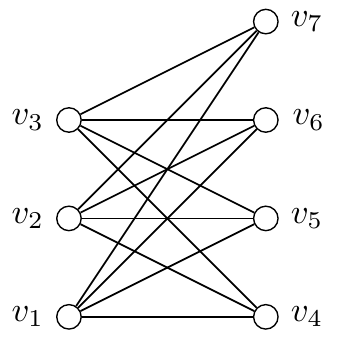}
		\caption{}
		\label{fig:g021}
		\centering
		\includegraphics[width=0.9\textwidth]{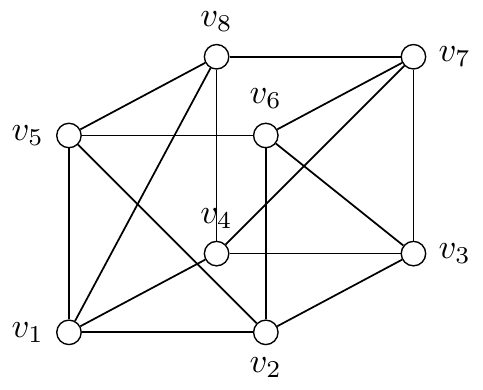}
		\caption{}
		\label{fig:g020}
	\end{subfigure}
	
	\caption{(a) graph $ G_e $; (b) graph $ G_f $; (c) graph $ G_g $}
	\label{fig:5}
\end{figure}

An arc is an edge treated as directed. A graph $ G $ is \emph{arc symmetric} if and only if all of its arc nailed graphs $ G(v_i; v_j) $, with $ (v_i, v_j)\in E $, are isomorphic to each other under the constraint $ \{\langle v_i,v_{i'} \rangle, \langle v_j, v_{j'}\rangle \} $, which means that the vertex $ v_i $ of the edge $ (v_i,v_j) $ can only be mapped to $ v_{i'} $ of the edge $ (v_{i'}, v_{j'}) $ and $ v_j $ can only be mapped to $ v_{j'} $ in any isomorphism between $ G(v_i;v_j) $ and $ G(v_{i'}; v_{j'}) $. The difference between arc symmetry and edge symmetry is that we treat the trigger edge as directed in the arc symmetry. There exists graphs that is both vertex symmetric and edge symmetric but not arc symmetric, e.g., the Doyle-Holt graph, see Figure 3.3 in \cite{godsil2001algebraic}.

\begin{lemma}\label{sym-arc}
	Given $ G=(V,E) $, if $ G $ is arc symmetric, then it is both vertex symmetric and edge symmetric.
\end{lemma}

In the literature, vertex symmetric is equivalent to vertex transitive and $ 0 $-arc transitive, edge symmetric is equivalent to edge transitive, and arc symmetric is equivalent to arc transitive and $ 1 $-arc transitive. We prefer the word `symmetric' more than `transitive', as the former is shorter and more intuitive.

Let's point out the levels of symmetry of a graph: 
No pair of isomorphic vertex nailed subgraphs and no pair of isomorphic edge nailed subgraphs $ \Longrightarrow $ Some pairs of isomorphic vertex nailed subgraphs or some pairs of isomorphic edge nailed subgraphs $ \Longrightarrow $ Vertex symmetric or edge symmetric $ \Longrightarrow $ Arc symmetric.

Given two multi-sets $ A=\{a_1,a_2,\cdots,a_s \} $ and $ B=\{b_1,b_2,\cdots,b_t \} $, we say $ A $ \emph{matches} $ B $ if the sort of $ a_1,a_2,\cdots,a_s $ is equal to the sort of $ b_1,b_2,\cdots,b_t $. For instance, $ \{3,7,3,2,1 \} $ matches $ \{1,2,3,3,7 \} $ but does not match $ \{3,7,2,1,4 \} $. More generally, $ \{ \{5,3,3\},\{5,2,2,8\},\{1,4,2\},\{1,3,3\} \} $ matches $ \{ \{1,2,4\},\{1,3,\newline 3\},\{2,2,5,8\},\{3,3,5\} \} $ but does not match $ \{ \{5,3,3\},\{1,3,4\},\{5,2,2,2\},\{1,4,\newline 2\} \} $.

The \emph{vertex property} of a graph $ G=(V,E) $ is the multi-set $ P_V(G)=\{|V_i|:i=1,2,\cdots,n \} $, where $ n=|V| $. We say two graphs $ G $ and $ G' $ are \emph{of the same vertex property} if their vertex properties match. Similarly, the \emph{edge property} of a graph $ G=(V,E) $ is the multi-set $ P_E(G)=\{deg(v_i,v_j):(v_i,v_j)\in E \} $. We say two graphs $ G $ and $ G' $ are \emph{of the same edge property} if their edge properties match.

Given one collapse of $ G $, say $ G^{col}(v_k) $, which has $ l+1 $ layers, $ G^{col}(v_k,i) $ for $ i=0,1,\cdots, l $, the \emph{collapse tomography} of $ G^{col}(v_k) $ is an ordered list of $ l $ ordered pairs of the vertex property and the edge property, i.e., \[ C^{tom}(G,v_k)=[[P_V(G^{col}(v_k,i));P_E(G^{col}(v_k,i))]:i=1,2,\cdots,l ]. \] Two collapse tomographies \emph{match} if all properties in the ordered list match at the corresponding position. For instance, \[ [ [\{5,3,3\};\{5,2,2,8\} ], [\{1,4,2\}; \{1,3, 3\} ], [\{7,5\}; \{3,3\} ] ] \] matches \[ [ [ \{3,3,5\};\{2,2,5,8\} ],[\{1,2,4\}; \{1,3,3\}], [\{5,7\};\{3,3\} ] ] \] but not \[ [ [ \{3,3,5\};\{2,2,5,8\} ],[\{1,3,3\};\{1,2,4\} ], [\{5,7\};\{3,3\} ] ]. \]

The collapse tomography of Figure \ref{fig:g016} is \[ [ [\{0,2,2,2\};\{3,3,3\} ],[\{2,2,2\};\{3,3,3\}] ]. \]

Given a graph $ G $, the \emph{collapse pattern} of $ G $ is a multi-set of collapse tomographies of $ G $, i.e., \[ C^{pat}(G)=\{C^{tom}(G,v_i):i=1,2,\cdots,|V| \}. \]

Given a nailed graph $ G(v_a) $, let $ G^{col}(v_a,j) $, $ j=0,1,2,\cdots,l $, be $ l+1 $ layers of its collapse. The \emph{collapse pattern} of $ G(v_a) $ is an ordered list of multi-sets of collapse tomographies of $ G $, i.e., \[ C^{pat}(G(v_a))=[\{C^{tom}(G,v_i): v_i\in V^{col}(v_a,j) \}: j=0,1,2,\cdots l ]. \]

Given a multi-nailed graph $ G(T) $ and its collapse $ G^{col}(T) $, which has $ l+1 $ layers, the corresponding collapse pattern is \[ C^{pat}(G(T))=[\{C^{tom}(G,v_i): v_i\in V^{col}(T,j) \}: j=0,1,2,\cdots l ]. \]

If each vertex $ v $ of the graph $ G $ is labeled by $ L(v) $, then the \emph{collapse tomography under labeling} $ L $ is \[ C^{tom}(G,L,v_i)=[[L(v_i);P_V(G^{col}(v_k,i));P_E(G^{col}(v_k,i))]:i=1,2,\cdots,l ], \]
and the \emph{collapse pattern under labeling} $ L $ for a nailed graph $ G(v_a) $ is \[ C^{pat}(G(v_a),L)=[\{C^{tom}(G,L,v_i): v_i\in V^{col}(v_a,j) \}: j=0,1,2,\cdots l ]. \]

We leave the study of properties of the collapse pattern of normal graph, nailed graph and multi-nailed graphs as a future work. 

\begin{lemma}\label{iso-pat-l}
	Given $ G=(V,E) $ and $ G'=(V',E') $, if $ G $ is isomorphic to $ G' $, then the collapse pattern of $ G $ matches that of $ G' $.
\end{lemma}

\begin{conjecture}\label{iso-pat}
	Given $ G=(V,E) $ and $ G'=(V',E') $, $ G $ is isomorphic to $ G' $ if the collapse pattern of $ G $ matches that of $ G' $.
\end{conjecture}

\begin{lemma}\label{nail-iso-pat-l}
	Given $ G $, $ G' $, and their nailed graphs, $ G(v_a) $ and $ {G'}(v_b) $, if $ G(v_a) $ is isomorphic to $ {G'}(v_b) $ with the constraint $ \{\langle v_a,v_b \rangle \} $, then the collapse pattern of $ G(v_a) $ matches that of $ {G'}(v_b) $. 
\end{lemma}

\begin{conjecture}\label{nail-iso-pat}
	Given $ G $, $ G' $, and their nailed graphs, $ G(v_a) $ and $ {G'}(v_b) $, $ G(v_a) $ is isomorphic to $ {G'}(v_b) $ with the constraint $ \{\langle v_a,v_b \rangle \} $ if the collapse pattern of $ G(v_a) $ matches that of $ {G'}(v_b) $. 
\end{conjecture}

\begin{lemma}\label{sym-pat-l}
	Given $ G=(V,E) $, if $ G $ is vertex symmetric, then all $ |V| $ collapse tomographies of $ G $ match each other, i.e., the collapse pattern of a vertex symmetric graph consists of $ |V| $ equal collapse tomographies. 
\end{lemma}

\begin{conjecture}\label{sym-pat}
	Given $ G=(V,E) $, $ G $ is vertex symmetric if all $ |V| $ collapse tomographies of $ G $ match each other, i.e., the collapse pattern of a vertex symmetric graph consists of $ |V| $ equal collapse tomographies. 
\end{conjecture}

Conjecture \ref{sym-pat} is a corollary of Conjecture \ref{nail-iso-pat}. In the algorithm for graph isomorphism, we only use Conjecture \ref{sym-pat}, as it is quite true intuitively. 

Let's call a graph $ G $ \emph{vertex indistinguishable} if all of its $ |V| $ collapse patterns $ C^{pat}(G(v_i)) $, with $ v_i\in V $, are equivalent. The question corresponding to Conjecture \ref{sym-pat} is ``Is there any graph that is vertex indistinguishable but not vertex symmetric?''.

\begin{lemma}\label{edge-sym-pat-l}
	Given $ G=(V,E) $, if $ G $ is edge symmetric, then all $ |E| $ collapse patterns of edge nailed graphs $ G(v_i,v_j) $, with $ (v_i,v_j)\in E $, match each other. 
\end{lemma}

\begin{conjecture}\label{edge-sym-pat}
	Given $ G=(V,E) $, $ G $ is edge symmetric if all $ |E| $ collapse patterns of edge nailed graphs $ G(v_i,v_j) $, with $ (v_i,v_j)\in E $, match each other. 
\end{conjecture}

In order to have a similar conjecture for arc symmetry, we redefine the collapse pattern for the arc nailed graph $ G(v_a;v_b) $ as an ordered list of collapse patterns: \[ C^{pat}(G(v_a;v_b))=[C^{pat}(G(v_a,v_b)), C^{pat}(G[v_a,v_b](v_a)), C^{pat}(G[v_a,v_b](v_b)) ], \] where $ C^{pat}(G(v_a,v_b)) $ is the collapse pattern of the edge nailed graph $ G(v_a,v_b) $, $ C^{pat}(G[v_a,v_b](v_a)) $ is the collapse pattern of the nailed graph at $ v_a $ obtained by removing the edge $ (v_a,v_b) $ in $ G $ and nail the vertex $ v_a $, and similarly for $ C^{pat}(G[v_a,v_b](v_b)) $.

\begin{lemma}\label{arc-sym-pat-l}
	Given $ G=(V,E) $, if $ G $ is arc symmetric, then all $ 2|E| $ collapse patterns of arc (an edge treated as directed) nailed graphs $ G(v_i;v_j) $, with $ (v_i,v_j)\in E $ and the direction is from $ v_i $ to $ v_j $, match each other. 
\end{lemma}

\begin{conjecture}\label{arc-sym-pat}
	Given $ G=(V,E) $, $ G $ is arc symmetric if all $ 2|E| $ collapse patterns of arc nailed graphs $ G(v_i;v_j) $ match each other. 
\end{conjecture}

The quantity, collapse pattern, contains a lot of information of a graph, so it may be a good argument to distinguish graphs. Up to now, our conjectures have never failed. Although we have tried our best to prove them, we cannot. Thus, we hope of finding a counterexample by a computer. 

Now let's introduce the dual graph. Given a graph $ G=(V,E) $, its \emph{dual graph} is $ \overline{G}=(V,\overline{E}) $, where $ (v_i,v_j)\in \overline{E} $ if and only if $ (v_i,v_j)\notin E $.

\begin{lemma}\label{dual}
	$ G $ is isomorphic to $ G' $ if and only if $ \overline{G} $ is isomorphic to $ \overline{G'} $.
\end{lemma}
\begin{proof}
	For any isomorphism $ \pi $ between $ G $ and $ G' $, $ (v_i,v_j) $ in $ G $ if and only if $ (\pi(v_i), \pi(v_j)) $ in $ G' $. So it is also true that $ (v_i,v_j) $ not in $ G $ if and only if $ (\pi(v_i), \pi(v_j)) $ not in $ G' $.
\end{proof}

When a graph of size $ n $ is vertex regular of a large degree $ w>\frac{n}{2} $, it seems easier to consider its dual graph, which is vertex regular of a smaller degree $ n-1-w<\frac{n}{2} $.

\begin{lemma}\label{comp}
	Given two complete graphs $ G $ and $ G' $, if $ G $ has $ n $ vertexes with vertex $ v_i $ labeled by $ L(v_i) $ and $ G' $ has $ n $ vertexes with vertex $ v'_j $ labeled by $ L(v'_j) $, then $ G $ is isomorphic to $ G' $ if and only if $ \{L(v_i):i=1,2,\cdots,n \} $ matches $ \{L(v'_j):j=1,2,\cdots,n \} $.
\end{lemma}
\begin{proof}
	The dual graph of a complete graph of size $ n $ is $ n $ independent vertexes.
\end{proof}

\subsection{Variations of Collapse Pattern}\label{var-pat}
In the remaining part of Section \ref{prepa}, we provide more discussions on our conjectures, which is not important for the next section. The reader may skip it if not interested.

In case that the above conjectures fail, we can vary the definitions of our collapse pattern to add more details of the graph.

Given a nailed graph $ G(v_a) $, let $ G^{col}(v_a,j) $, $ j=0,1,2,\cdots,l $, be $ l+1 $ layers of its collapse, and let $ G^{ex}(v_i;v_a,j) $ be the extension with a base point $ v_i\in V^{col}(v_a,j) $ and the base subgraph $ G^{col}(v_a,j) $. The \emph{varied collapse pattern} of $ G(v_a) $ is an ordered list of multi-sets of ordered collapse tomographies, i.e., 
\begin{align*}
	C^{pat}(G(v_a))=[\{[&C^{tom}(G,v_i);C^{tom}(G^{col}(v_a,j),v_i);\\
	&C^{tom}(G^{ex}(v_i;v_a,j),v_i) ]: v_i\in V^{col}(v_a,j) \}: j=0,1,\cdots l ]. 
\end{align*}

This varied definition encodes more information of the nailed graph.

For an ordinary graph, we also change the definition as \[ C^{pat}(G)=\{C^{pat}(G(v_i)): v_i\in V \}. \]

Let $ G(T) $ be a multi-nailed graph with nailed vertexes $ T $ and its collapse $ G^{col}(T) $, which has $ l+1 $ layers. The corresponding \emph{varied collapse pattern} of $ G(T) $ is 
\begin{align*}
C^{pat}(G(T))=[\{[&C^{tom}(G,v_i);C^{tom}(G^{col}(T,j),v_i);\\
&C^{tom}(G^{ex}(v_i;T,j),v_i) ]: v_i\in V^{col}(T,j) \}: j=0,1,\cdots l ]. 
\end{align*}

\begin{lemma}
	All of our definitions of collapse tomography and collapse pattern are well-defined and computable in polynomial time.
\end{lemma}

If $ G $ of size $ n $ is a graph constructed by linking $ n-1 $ vertexes to one vertex $ v $, we call it a \emph{diverging graph} and $ v $ is called the \emph{source} of the graph..

\begin{lemma}\label{diverg}
	If $ G $ is a diverging graph with the source $ v $ and $ n-1 $ vertexes labeled by $ L(v_i) $, and $ G' $ is another diverging graph with source $ v' $ and $ m-1 $ vertexes labeled by $ L(v'_j) $, then $ G(v) $ is isomorphic to $ G'(v') $ if and only if $ \{L(v_i):v_i\in V \} $ matches $ \{L(v'_j):v'_j\in V' \} $.
\end{lemma}

Suppose we want to define for a nailed graph $ G(v_a) $ a quantity that is useful to distinguish graphs and easy to be proved by induction, let such quantity be $ C^q(G(v_a)) $. We hope that two graphs are isomorphic if and only if their values of the quantity are equivalent.

\begin{theorem}
	Given $ G(v_a) $ and $ {G'}(v_b) $, $ G(v_a) $ is isomorphic to $ {G'}(v_b) $ with the constraint $ \{\langle v_a,v_b \rangle \} $ if and only if $ C^q(G(v_a)) $ matches $ C^q({G'}(v_b)) $, if we define the quantity as \[ C^q(G(v_a))=\{C^q(G^{ex}(v_i,v_a)):v_i\in V_a \}, \] where $ G^{ex}(v_i,v_a) $ is the extension based on $ v_i $ linking to $ v_a $ and the base subgraph consisting only of the edge $ (v_i,v_a) $.
\end{theorem}
\begin{proof}
	Let $ L(v_i)=C^q(G^{ex}(v_i,v_a)) $, then by induction and Lemma \ref{diverg}, we can prove this theorem.
\end{proof}

\begin{corollary}
	If $ C^q(G(v_a)) $ is a quantity for nailed graphs, then $ C^q(G)=\{C^q(G(v_i)):v_i\in V \} $ is a quantity for normal graphs.
\end{corollary}

\begin{lemma}\label{norm-lab}
	If $ C^q(G)=\{C^q(G(v_i)):v_i\in V \} $ is a quantity for normal graphs, then $ C^q(G,L)=\{[C^q(G(v_i)),L(v_i)]:v_i\in V \} $ is a quantity for the case that each vertex is labeled by $ L(v_i) $.
\end{lemma}

\begin{lemma}\label{layer-1}
	If $ G(v_a) $ and $ G'(v_b) $ are two nailed graphs with their collapses $ G^{col}(v_a) $ and $ {G'}^{col}(v_b) $ having only Layer $ 0 $ and Layer $ 1 $, then $ G(v_a) $ is isomorphic to $ G'(v_b) $ if and only if $ G^{col}(v_a,1) $ is isomorphic to $ {G'}^{col}(v_b,1) $.
\end{lemma}

\begin{theorem}
	Given $ G(v_a) $ a nailed graph, if we define the quantity as \[ C^q(G(v_a))=\{[C^q(G(v_i,V_a)),C^q(G^{ex}(v_i,V_a\cup \{v_a\}))]:v_i\in V^{col}(v_a,1) \}, \] where $ G(v_i,V_a) $ is the graph $ G^{col}(v_a,1) $ nailed at $ v_i $, and $ G^{ex}(v_i,V_a\cup \{v_a\}) $ is the extension based on $ v_i $ with the base subgraph consisting of vertexes $ V_a $ and $ v_a $, then this quantity can distinguish graphs according to isomorphism.
\end{theorem}
\begin{proof}
	Let $ L(v_i)=C^q(G^{ex}(v_i,V_a\cup \{v_a\})) $, then by induction and \Cref{norm-lab} and \Cref{layer-1}, we can prove this theorem.
\end{proof}

As the lack of details of a graph, we cannot accelerate the graph isomorphism problem in the abstract approach. In the next section, we apply a bunch of methods for the acceleration of our algorithm.


\section{Algorithms}\label{algo}

Now, we are ready to present the algorithm for determining whether two graphs are isomorphic.

Note that our algorithm is only based on \Cref{sym-pat} and no other assumption not proved. If we remove this conjecture, we can show its quasi-polynomial time efficiency.\\

Our algorithm for graph isomorphism is $ GI(G,G')=GI(G,G',L_0) $, where $ L_0=\varnothing $.

\noindent\textbf{Algorithm: Graph Isomorphism $ GI(G,G',L_k) $ with Labels }

\begin{enumerate}[leftmargin=4\parindent]
	\item[\textbf{Input:}] Two graphs $ G=(V,E) $ and $ G'=(V',E') $, without self-loop or multi-edge. $ V=\{v_1,v_2,\cdots,v_n \} $, $ V'=\{v'_1,v'_2,\cdots,v'_n \} $. Every vertex $ v\in V\cup V' $ is labeled by $ L_k(v) $.
\end{enumerate} 

\begin{enumerate}[leftmargin=4\parindent]
	\item[\textbf{Output:}] If $ G $ and $ G' $ are isomorphic under the labeling function $ L_k $, output `Yes', otherwise `No'.
\end{enumerate} 

\begin{enumerate}[leftmargin=4\parindent]
	\item[\textbf{Runtime:}] $ O(n^3) $ with \Cref{sym-pat}; $ O(n^{1.5\log n}) $ without \Cref{sym-pat}. The worst case is when the input is two vertex indistinguishable but not vertex symmetric graphs.
\end{enumerate}

\noindent\textbf{Procedure:} 

\begin{enumerate}[leftmargin=2\parindent,label=\textbf{\arabic*}.,ref=\arabic*]
	\item Compute $ C^{tom}(G,L_k,v_i) $ for all $ v_i\in V $, and $ C^{tom}(G',L_k,v'_j) $ for all $ v'_j\in V' $. If they all match each other, then by \Cref{sym-pat}, $ G $ and $ G' $ are vertex symmetric. Then if $ G $ and $ G' $ are regular of degree $ w\leq \frac{n}{2} $, we call the sub-algorithm $ GI(G,G',C,L_k) $, with $ C=\{\langle v_1,v'_1\rangle \} $. If $ G $ and $ G' $ are regular of degree $ w> \frac{n}{2} $, we call the sub-algorithm $ GI(\overline{G},\overline{G'},C,L_k) $, with $ C=\{\langle v_1,v'_1\rangle \} $. Output `Yes' iff the sub-algorithm returns `Yes'. Here, our target is to find an isomorphism. Note that, if we do not appeal to \Cref{sym-pat}, then we have to check the isomorphism for all $ 2n $ nailed graphs of $ G $ and $ G' $ in the worst case, instead of only one pair. We may use the method `try-and-error' to further accelerate the algorithm, as we believe that there is a non-negligible probability for finding a pair of isomorphic nailed graphs if two graphs are vertex indistinguishable but not vertex symmetric. Because of the use of dual graph, the size of the next occurrence of vertex indistinguishable graph is no larger than $ \frac{n}{2} $. 
	\item Compute $ C(G)=\{C^{pat}(G(v_i),L_k): v_i\in V \} $, and $ C(G')=\{C^{pat}({G'}(v'_j),\newline L_k): v'_j\in V' \} $.  Sort these two multi-sets of collapse patterns. If they do not match, output `No'. Otherwise, suppose the rarest collapse pattern in $ C(G) $ is $ C^{pat}(G(v_x),L_k) $ and the corresponding one in $ C(G') $ is $ C^{pat}({G'}(v'_y),L_k) $, let $ \beta=\{v_i:C^{pat}(G(v_i),L_k)=C^{pat}(G(v_x),L_k) \} $ and $ \beta'=\{v'_j:C^{pat}({G'}(v'_j),L_k)=C^{pat}({G'}(v'_y),L_k) \} $. The base subgraph of $ G $ is $ G^{\beta} $ and that of $ G' $ is $ {G'}^{\beta'} $. The extensions of $ G^{\beta} $ are $ Ex(\beta)=\{G^{ex}(v_i,\beta):v_i\in V^{\beta}\} $ and those of $ {G'}^{\beta'} $ are $ Ex(\beta')=\{{G'}^{ex}(v'_j,\beta'):v'_j\in {V'}^{\beta'}\} $. Call the sub-algorithm $ GI(G_1^{ex}(v_1,\beta_1), G_2^{ex}(v_2,\beta_2), C,L_k) $ for all pairs of extensions $ (G_1^{ex}(v_1,\beta_1), G_2^{ex}(v_2,\beta_2)) $ in the union set $ Ex(\beta)\cup Ex(\beta') $ with $ C=\{\langle v_1,v_2\rangle \} $. Next we assign two extensions with the same label iff they are isomorphic. The labeling function $ L_{k+1} $ assigns two base points with the same label iff $ L_k $ assigns them with the same label and the labels for their extensions are of the same. Then we call the sub-algorithm $ GI(G^{\beta},{G'}^{\beta'}, L_{k+1} ) $, with vertex $ v\in \beta\cup \beta' $ labeled by $ L_{k+1}(v) $. Output `Yes' iff this sub-algorithm returns `Yes'.

\end{enumerate}

\noindent\textbf{Algorithm: Graph Isomorphism $ GI(G, G',C,L_k) $ with Constraint and Labels }

\begin{enumerate}[leftmargin=4\parindent]
	\item[\textbf{Input:}] Two graphs $ G=(V,E) $ and $ G'=(V',E') $, without self-loop or multi-edge. $ V=\{v_1,v_2,\cdots,v_n \} $, $ V'=\{v'_1,v'_2,\cdots,v'_n \} $. The constraint is $ C=\{\langle v_a,v'_b\rangle\} $, where $ v_a $ is the nailed vertex of $ G(v_a) $ and $ v'_b $ is the nailed vertex of $ G'(v'_b) $. Every vertex $ v\in V\cup V' $ is labeled by $ L_k(v) $.
\end{enumerate} 

\begin{enumerate}[leftmargin=4\parindent]
	\item[\textbf{Output:}] If $ G(v_a) $ and $ G'(v'_b) $ are isomorphic under the constraint $ C $ and the labeling function $ L_k $, then output `Yes'; otherwise `No'.
\end{enumerate} 


\noindent\textbf{Procedure:} 

\begin{enumerate}[leftmargin=2\parindent,label=\textbf{\arabic*}.]
	\item Compute $ C^{pat}(G(v_a),L_k) $ and $ C^{pat}({G'}(v'_b), L_k) $. If they do not match, then output `No'. Otherwise, let the collapse $ G^{col}(v_a) $ have $ l+1 $ layers. 
	\item If $ l=1 $, call the sub-algorithm $ GI(G^{col}(v_a,1),{G'}^{col}(v'_b,1),L_k) $. Output `Yes' iff this sub-algorithm returns `Yes'. Here, the input graphs may be again vertex indistinguishable but not vertex symmetric.
	\item For $ l\geq 1 $, from $ x=l-1 $ to $ x=0 $, we update the labels of vertexes in layer $ x $ with the layer $ x+1 $ by the following procedure. Let $ \beta(x) $ be the set of vertexes in layers from $ 0 $ to $ x $ of $ G^{col}(v_a) $ and $ \beta'(x) $ be the set of vertexes in layers from $ 0 $ to $ x $ of $ {G'}^{col}(v'_b) $. Let $ Ex(\beta(x))=\{G^{ex}(v_i,\beta(x)):v_i\in V^{col}(v_a,x)\} $ and $ Ex(\beta'(x))=\{{G'}^{ex}(v'_j,\beta'(x)):v'_j\in {V'}^{col}(v'_b,x) \} $. Suppose the labeling function for layer $ x $ is $ L_{k,x} $ with $ L_{k,l}=L_k $. Call the sub-algorithm $ GI(G_1^{ex}(v_1,\beta_1(x)), G_2^{ex}(v_2,\beta_2(x)), C, L_{k,x+1}) $ for all pairs of extensions $ (G_1^{ex}(v_1,\beta_1(x)), G_2^{ex}(v_2,\beta_2(x))) $ in the union set $ Ex(\beta(x))\cup Ex(\beta'(x)) $ with $ C=\{\langle v_1,v_2\rangle \} $. Usually, we omit the labels of base points when we are determining isomorphism of their extensions. Next we assign two extensions with the same label iff they are isomorphic. The labeling function $ L_{k,x} $ assigns two base points with the same label iff $ L_{k} $ assigns them with the same label and the labels for their extensions are of the same. On the whole, we distinguish two nailed graphs, according to isomorphism, layer-by-layer.
	
\end{enumerate}

\begin{theorem}
	The algorithm $ GI(G,G',L_k) $ can correctly distinguish graphs according to isomorphism and is of polynomial time if \Cref{sym-pat} is true.
\end{theorem}

\begin{theorem}
	If \Cref{sym-pat} is not true, then the algorithm $ GI(G,G',L_k) $ is of quasi-polynomial time $ O(n^{1.5\log n}) $.
\end{theorem}
\begin{proof}
	The recursion is $ T(n)=n n^2 T(\frac{n}{2}) $ for the worst case, which has the most chance of encounter of graphs that is vertex indistinguishable but not vertex symmetric.
\end{proof}

Suppose now we allow graphs with self-loop and multi-edge. First, we remove all self-loop and multi-edge, and by the algorithm $ GI(G,G') $, we can find an isomorphism. Then we can check this isomorphism for graphs with self-loop and multi-edge. Another method is to use the labeling procedure to remove self-loop and multi-edge before the algorithm for graphs without those. 

With the algorithm for isomorphism, we can construct an algorithm for automorphism without too much effort.

\section{Conclusion}
In this paper, we provide an algorithm for the graph isomorphism problem. We have shown that the graph isomorphism problem is based on the recognition problem of a vertex symmetric graph. \Cref{sym-pat} can solve the latter problem very efficiently. Without this conjecture, our algorithm is of quasi-polynomial time. It is possible to further improve our algorithm. 

Although some techniques in this article might have been already known before, we apologize for the possibility of not pointing out, for our limited knowledge. We plan to write a detailed and also interesting review paper to make it clear and talk all aspects of the graph isomorphism problem, including the research history and its applications. 

Left works: What is the algorithm for directed graphs? Experimental benchmark of algorithms for graph isomorphism. Problems relating to the graph isomorphism. Properties of collapse tomography and collapse pattern. A more refined analysis of our algorithm. Find a counter-example to \Cref{sym-pat} and other conjectures.

\bibliography{fang}
\end{document}